%% file: improved_sparsifier_variants.tex
\documentclass[11pt,letterpaper]{article}

\usepackage{mathrsfs}           
\usepackage{vmargin,fancyhdr}   
\usepackage{algorithm}
\usepackage{algorithmic}

\usepackage{enumerate}

\usepackage{amsmath,amssymb}    
\usepackage{verbatim}           
\usepackage{xspace}             
\usepackage{graphicx,float}     
\usepackage{ifthen,calc}        
\usepackage{textcomp}           
\usepackage{fancybox}           
\usepackage{hhline}             
\usepackage{float}              

\usepackage[rflt]{floatflt}     
\usepackage[small,compact]{titlesec}
\usepackage{setspace}
\usepackage{subfigure}

\usepackage{color}
\definecolor{ForestGreen}{rgb}{0.1333,0.5451,0.1333}
\usepackage[dvips]{thumbpdf}

\setpapersize{USletter}
\setmarginsrb{.75in}{.5in}        
             {.75in}{.5in}        
             {.25in}{.25in}     
             {.25in}{.5in}      
\newcommand{\showccc}[0]{0}
\newcommand{\ccc}[2][nothing]{
  \ifthenelse{\showccc=0}{}{
    \ensuremath{^{\Lsh\Rsh}}\marginpar{\raggedright\tiny\textsf{%
        \ifthenelse{\equal{#1}{nothing}}{}{\textbf{#1}\\}#2}}}}

\newcounter{hours}\newcounter{minutes}
\newcommand{\hhmm}{%
  \setcounter{hours}{\time/60}%
  \setcounter{minutes}{\time-\value{hours}*60}%
  \ifthenelse{\value{hours}<10}{0}{}\thehours:%
  \ifthenelse{\value{minutes}<10}{0}{}\theminutes}
\ifthenelse{\showccc=0}{\rhead{}}{\rhead{\today \ [\hhmm]}}
\newtheorem{theorem}{Theorem}[section]

\newtheorem{definition}[theorem]{Definition}

\newtheorem{lemma}[theorem]{Lemma}

\newcommand{\softO}{\tilde{O}}
\newcommand{\up}{5}
\newcommand{\st}{\mathrm{stretch}}
\newcommand{\sff}{s_f}
\newcommand{\tff}{t_f}

\newenvironment{proof}
  {\vskip5pt\hspace*{15pt}{\it Proof}.\hskip10pt}{\qed\vskip5pt}
\def\qed{\nobreak\hskip1pt~$\;\;\scriptstyle\Box$}

\floatstyle{ruled}
\newfloat{algo}{tbp}{lop}
\floatname{algo}{Algorithm}

\DeclareMathOperator{\trace}{trace}

\DeclareMathOperator{\ccap}{cap}

\usepackage{float}
\floatstyle{plain}\newfloat{myfig}{t}{figs}[section]
\floatname{myfig}{\textsc{Figure}}
\floatstyle{plain}\newfloat{myalg}{H}{algs}[section]
\floatname{myalg}{}
\begin{document}

\title{Faster spectral sparsification \\ and numerical algorithms for SDD matrices}
\author{Ioannis Koutis \\
CSD-UPRRP \\ ioannis.koutis@upr.edu   \and Alex Levin  \\ MATH-MIT  \\ levin@mit.edu \and Richard Peng \\  CSD-CMU \\yangp@cs.cmu.edu}

\maketitle
\maketitle
\input{abstract}
\input{intro}

\input{background}

\input{sparsification}

\input{improve}

\input{discussion}

\bibliographystyle{plain}

\input{improved_sparsifier_variants.bbl}
\end{document}

%% file: abstract.tex
\begin{abstract}
We study algorithms for spectral graph sparsification. The
input is a graph $G$ with $n$ vertices
and $m$ edges, and the output is a {sparse} graph $\tilde{G}$ that approximates $G$
in an algebraic sense. Concretely, for all vectors $x$ and any $\epsilon>0$,
the graph $\tilde{G}$ satisfies
$$
     (1-\epsilon) x^T L_G x  \leq x^T L_{\tilde{G}} x \leq (1+\epsilon) x^T L_G x,
$$
where $L_G$ and $L_{\tilde{G}}$ are the Laplacians of $G$ and $\tilde{G}$ respectively.

We show that the fastest known
algorithm for computing a sparsifier with $O(n\log n/\epsilon^2)$ edges
can actually run in $O(\sff m\log^2 n)$ time\footnote
{We denote by $\sff$ and $\tff$ the $O(\log \log n)$ factors that appear
in the stretch and time guarantees of the best currently known algorithm
for computing low-stretch trees \cite{AbrahamN12}.
 It is conjectured that $s_f=O(1)$ and $t_f=O(1)$ is possible.},
 an $O(\log n)$ factor faster than before.
 We also present faster sparsification
 algorithms for slightly dense graphs.  Specifically, we give an algorithm that
runs in $O(\tff m\log n)$  time and generates a sparsifier with $O(\sff n\log^3{n}/\epsilon^2)$ edges.
We also give an $O(m\log \log n)$ time algorithm
for graphs with more than $\sff^2 n\log^\up n \log \log n$ edges and an $O(m)$ algorithm
for graphs with more than $\tff^2 \sff^3 n\log^{10} n  $ edges and unweighted graphs
with more than $\sff^3 n\log^8 n  $ edges.
The improved sparsification algorithms are employed to accelerate
linear system solvers and algorithms  for computing fundamental eigenvectors of slightly dense SDD matrices.

\end{abstract}

%% file: intro.tex
\section{Introduction}

The efficient transformation of dense instances of
graph problems to nearly equivalent sparse instances is a
very powerful tool in algorithm design.
The idea, widely known
as {\em graph sparsification},  was originally
introduced by Bencz{\'u}r and
Karger \cite{bk} in the context of cut problems. Spielman
and Teng \cite{SpielmanTeng04} generalized the {\em cut-preserving
sparsifiers} of Bencz{\'u}r and
Karger to the more powerful {\em spectral sparsifiers},
which preserve in an algebraic sense the Laplacian
matrix of the dense graph. The main motivation of spectral
sparsifiers was the design of nearly-linear time algorithms for the
solution of symmetric diagonally dominant (SDD) linear systems. A matrix $A$
is SDD if it is symmetric and for all $i$, $A_{ii} \geq \sum_{j\neq i} |A_{ij}|$.

Bencz{\'u}r and Karger  proved that,
 for arbitrary $\epsilon,$ cuts can
be preserved within a factor of $1\pm\epsilon$ by a graph
with $O(n \log n/\epsilon^2)$ edges. This graph
 can be computed
by a randomized algorithm that runs in $O(m\log^3 n)$ time\footnote{All sparsification algorithms in this paper
are randomized with a probability failure inversely proportional to $n$. They consist of a preprocessing phase followed by the generation of the sparsifier which in general
can be performed in time proportional to the number of edges in it (e.g. $O(n\log n/\epsilon^2)$).  For the sake of conciseness our running time statements will include only the time for preprocessing and will omit the failure
probability.},
where $m$ is the number of edges in the dense graph.
Spielman and Teng gave the first construction
of spectral sparsifiers,
but   the edge count  of these objects was
several log factors bigger than  that of Bencz{\'u}r and
Karger's cut-preserving sparsifiers.
However, recent progress that we review below allows for
the construction of spectral sparsifiers
with $O(n\log n/\epsilon^2)$ edges in
${O}(\sff m\log^3 n)$ time.

Sparsification can be employed to immediately accelerate
algorithms for numerous problems. In several cases and depending
on the density of the instance, the sparsification
routine dominates the running time of the sparsifier-enhanced
algorithm. This is a strong incentive for speeding up the construction
of sparsifiers even further.

This problem was  undertaken in the context of cut-preserving sparsifiers
by Fung {\it et al.}~\cite{FungHHP11}. Improving upon the work of
Bencz{\'u}r and Karger, they proved that there is an $O(m\log^2 n)$ time
algorithm that computes a sparsifier with $O(n \log n/\epsilon^2)$ edges.
This stands as the fastest known algorithm with this sparsity guarantee for
general graphs. However, Fung \textit{et al.}\ also showed that we can do even
better on slightly more dense graphs. More concretely,
they proved that there is an $O(m)$ time algorithm that computes
a sparsifier with $O(n \log^2 n/\epsilon^2)$ edges.
Note  that by transitivity, a combination of the two algorithms can
produce a graph with $O(n \log n/\epsilon^2)$ edges in $O(m + n \log^4 n)$ time.
In other words, there is a linear time sparsification algorithm for graphs with more
than $n\log^4 n$ edges.

This leads us to the main question we address in this paper:
Is something analogous possible for spectral sparsification?
We answer the question in the affirmative. We first show that a slight modification of the
Spielman-Srivastava algorithm \cite{effres} can improve the run time to $O(\sff m\log^2 n )$.
This nearly matches the general case algorithm of~\cite{FungHHP11}.
We present three additional
sparsification algorithms.
The first is a variation of the Spielman-Srivastava algorithm
that generates a sparsifier with $O(\sff n\log^3 n/\epsilon^2)$
edges in $O(\tff m\log n)$ time.
The second  produces  a sparsifier with
$O(n\log n/\epsilon^2)$ edges in
$O(m\log \log n)$ time, assuming the input has more than $\sff^2 n\log^\up n \log \log n $ edges.
The third produces a sparsifier with
${O}(n\log n/\epsilon^2)$ edges in $O(m)$ time
assuming the input is unweighted and has more than $\sff^3 n\log^8 n$ edges,
or if it has more than $\tff^2 \sff^3 n\log^{10} n $ in the weighted case.

\subsection*{Applications in numerical algorithms}

The $(1\pm\epsilon)$-sparsifiers we obtain can be employed
in a standard way  as preconditioners for SDD linear systems,
giving us faster solvers for slightly dense graphs: (i)
an $O(m\log \log n)$ time solver for systems with more than $\sff^2 n\log^\up n \log \log n$
non-zero entries and (ii) an $O(m)$ time solver
for Laplacians of graphs with more than $\tff^2\sff^3 n\log^{10} n$
non-zero entries. The best previously known algorithm \cite{kmp11} runs in
$O(\sff m \log n \log(1/\delta))$ time.

In addition, our sparsification algorithms accelerate the computation of an approximate Fiedler eigenvector
of a graph Laplacian $L_G$. An $(1+ \epsilon)$-approximate eigenvector is a unit norm vector $x$ such that $x^T L_G x$
is within a factor $1+ \epsilon$ from the eigenvalue $\lambda_2$ of $L_G$.
The algorithm consists of two steps: (i)~computing a spectral sparsifier $\tilde{G}$ that $(1\pm \epsilon/2)$-approximates
the input graph $G$ and
 (ii)~computing a $(1+ \epsilon/3)$-approximate eigenvector of $\tilde{G}$; this will automatically be an $(1\pm \epsilon)$-approximate eigenvector of the
(more) dense input graph because the spectral sparsification step preserves the eigenvalues of $G$ within $1\pm \epsilon/2$.
Hence combining our sparsification algorithms with the inverse power method \cite{stlinsolve} (which consists
of solving $O(\log n \log(1/\epsilon))$ systems in $L_{\tilde{G}}$)
gives an approximate eigenvector in
$O(m + \sff^2 n\log^5 n\log(1/\epsilon)/\epsilon^2)$ time.
The fastest previously  known algorithm runs in time $O(\sff m \log^2 n \log(1/\epsilon))$.
The same result applies to the computation of the Fiedler eigenvector of a normalized Laplacian $D^{-1/2} L_G D^{-1/2}$; applying the inverse power method on $D^{-1/2} L_{\tilde{G}} D^{-1/2}$ gives the required eigenvector.


 We note
here that one practical application of eigenvectors is in partitioning
algorithms;  the analysis of Cheeger's inequality \cite{chung07}
tells us how to turn an approximate
Fiedler vector into  a partition.
Hence, we give an improvement   to the running time
of  a fundamental graph partitioning algorithm.
Finally we note that the computation of additional eigenvectors
can be performed in the same
amount of time (per vector) by restricting the action
of the matrix to the complement of the subspace spanned by the
previously computed eigenvectors.

\section{Overview of our techniques}

\subsection{Brief background on spectral sparsification}

The first algorithm for edge-efficient spectral
sparsifiers was given by Spielman and Srivastava \cite{effres}.
Their algorithm produces a sparsifier with $O(n\log n/\epsilon^2)$
edges in a very elegant way:  it samples
edges with replacement.
The  probability of sampling an
edge is proportional to
its weight multiplied by  its effective resistance
in the resistive electrical network associated
with the given graph.

Computing the effective resistance of a given edge
requires---almost by definition---the solution
of a linear system on the graph Laplacian.\footnote{Laplacian matrices
are SDD.}  However, Spielman and Srivastava
also provided a way of estimating {\em all} $m$
effective resistances, via solving $O(\log n)$ SDD linear systems.
This holds under the assumption that the SDD solver is direct,
i.e.\ it outputs an exact solution. The use of a nearly-linear
time {\em iterative} solver that computes approximate
solutions introduces an additional source of imprecision;
Spielman and Srivastava showed that solving the systems
up to an inverse polynomial precision is sufficient
for sparsification. This brings the running time
of their algorithm to $O(\sff m\log^{c+2} n)$, where
$c$ is the constant appearing in the running time
of the SDD solver.

\subsection{The $O(\sff m\log^2 n)$ time algorithm}
While the work of Spielman and Srivastava did not
improve the running time of the SDD solver, it proved to be
a decisive step towards the fast SDD solver of
Koutis, Miller, and Peng \cite{kmplinsolve,kmp11},
which  runs in time
$O(\sff m \log n \log(1/\delta))$, where
$\delta$ is the desired precision. Using
this solver in the Spielman and Srivastava sparsification
sampling scheme immediately yields an
$O(\sff m\log^{3} n/\epsilon^2)$ time algorithm.
This brings us to the first contribution of this paper, a tighter analysis
of the Spielman and Srivastava algorithm. In Section \ref{sec:improve}
we show that solving the systems up to  \emph{fixed} precision is actually sufficient
for sparsification. This decreases the running time to $O(\sff m\log^2 n/\epsilon^2)$.

\subsection{Faster algorithms: The main idea}

To get our two faster algorithms, we will trade accuracy
in the computation of effective resistances for speed.
The idea is to transform the input graph $G$ into another
graph $H$ where  effective resistances can be computed faster while still
providing good  bounds for the true effective resistances in $G$.
These approximate effective resistances can still be used
for sparsification at the expense of additional sampling
\cite{kmplinsolve} that yields slightly more dense sparsifiers.
These sparsifiers can be re-sparsified to $O(n\log n/\epsilon^2)$
edges by applying the fast general-case algorithm.

\subsection{The $O(\tff m\log n)$ time algorithm}

The $O(\tff m\log n)$ time algorithm is based on the
observation that the Spielman-Srivastava scheme can be
implemented to run in $O(\tff m\log n)$ time on a {\em spine-heavy} approximation
$H$ of $G$. The spine-heavy graph $H$ is derived in $O(\tff m\log n)$ time
from $G$ by computing a low-stretch tree of $G$ and scaling it up by a
$O(\sff \log^2 n)$ factor. In \cite{kmp11} it was shown that
linear systems involving the Laplacian of $H$ can be solved in $O(m)$ time,
enabling the faster implementation of the Spielman-Srivastava scheme on $H$.
At the same time the effective resistances in $H$ are at
most a $O(\sff \log^2 n)$ factor smaller than those in $G$.
Sampling with respect to these estimates, allows us to get a sparsifier $\tilde{G}'$ with $O(\sff n\log^3 n/\epsilon^2)$
edges.  Re-sparsifying $\tilde{G}'$  gives a sparsifier $\tilde{G}$ with $O(n \log n/\epsilon^2)$ edges
in $O(\sff^2 n\log^5 n/\epsilon^2)$ time.
 The details are given in Section~\ref{sec:improve}.

\subsection{The $O(m\log \log n)$ and $O(m)$ time algorithms}

There are two major bottlenecks in the $O(\tff m\log n)$ time algorithm. In order
to work around them, we introduce several ideas of independent interest.

\medskip
{\bf (a)} The first bottleneck is in the computation of the low-stretch tree;
all known algorithms for computing a low-stretch tree run in time at least $O(m \log n)$.
The solution to this problem involves two steps. We first observe that
if we can settle for
 a weaker stretch guarantee, it is enough to find a low-stretch
tree of a {\em subgraph} $H$ of the input graph $G$ which preserves {\em all} cuts of $G$ within a
polylogarithmic factor. This allows the existing low-stretch tree computation to run faster,
assuming the subgraph is sparser by an $O(\tff \log n)$ factor.
The subgraph $H$ can be viewed as an approximate cut sparsifier of $G$, a notion
very similar to the incremental spectral sparsifiers from \cite{kmplinsolve}; these
are computed by first finding a low-stretch tree and sampling off-tree edges
with probability proportional to their stretch over the tree. Inspired by this
idea, we give an even simpler algorithm for computing the graph $H$:
we find a maximum weight spanning tree and then we sample uniformly the off-tree edges.
The proof that a variant of this simple procedure returns the desired incremental cut sparsifier relies
on Karger's earlier work on cut sparsification \cite{Karger98}.

{\bf (b)} Having removed the low-stretch tree computation obstacle we can now
attempt to mimic the steps of the $O(\tff m \log n)$ time algorithm. In fact
it is possible to take care of the system-solving part of the Spielman-Srivastava scheme in $O(m)$
time by merely scaling-up the low-stretch tree by a larger factor. However
this is not enough; there is still a bottleneck that lies in the computations
{after} the solution of the linear systems in the heavier-spine graphs.
These are $m$ simple manipulations of vectors of dimension $O(\log n)$. In an earlier version of this work,
we attempted to work around this problem by reducing the dimension of these vectors, but
at a significant loss of sparsity \cite{KoutisLP12}.

Here we take a different route. We first  use the low-stretch tree to construct
an {\em approximate sparsifier} i.e.\
a sparse approximation for the input graph, but of moderate quality.
The significant departure from the Spielman-Srivastava
scheme comes in the next step which computes estimates
of the effective resistances using a combinatorial rather than
algebraic approach. Concretely, we observe that with some additional work we can `leverage'
the approximate sparsifier to compute a sparsifier for $G$ as well.

Indeed, let $H$ be a $\kappa$-approximation of $G$
(see Definition~\ref{defn:kappaapprox}) and
$\tilde{H}$ be a sparsifier of $H$ with $O(n\log n)$ edges;
$\tilde{H}$ is the approximate sparsifier.
Then we  generate a low-stretch spanning tree
$T$ of $\tilde{H}$ in $O(\sff n\log^2 n)$ time  and approximate
the effective resistances of $G$ over $T$ in $O(m)$ time.
We will be able to claim that these approximate values are
enough to generate a sparsifier $\tilde{G}'$ for $G$ with
$O(\sff n \kappa \log^3 n)$ edges. Finally from $\tilde{G}'$ we can compute a sparsifier $\tilde{G}$
with $O(n\log n/\epsilon^2)$ edges in $O(\sff^2 \kappa n\log^5 n)$ time
using our first algorithm.

We will derive our $O(m)$ algorithm via a single application
of the above `leveraging' idea for $\kappa = O(\sff^3 \log^5 n)$.
To improve over the $O(\sff m\log n)$ algorithm  for even sparser graphs
 we will progressively sparsify
a sequence of $t=O(\log \log n)$ graphs
$H={H}_0,{H}_1,\ldots,{H}_t=G$, such that ${H}_{i}$ is a $2$-approximation
of ${H}_{i+1}$: given the sparsifier
for ${H}_i$ we can construct the sparsifier for
${H}_{i+1}$ via the leveraging idea. The details are given in Section \ref{sec:fastest}.

%% file: background.tex
\section{Background on spectral graph theory and sparsification} \label{sec:background}

\subsection{The graph Laplacian and its pseudoinverse}
Let $G = (V,E,w)$ be an undirected weighted graph on  $n$ vertices,
which we identify with the integers
$\{1,2,\ldots, n\},$  and $m$
edges, where the
weight of edge $e$ is given by $w_e.$ Without loss of generality
we will assume that minimum weight is $1$.
We will also assume that matrices discussed
 below are represented as
adjacency lists.

The Laplacian of $G$ is denoted by $L_G.$ It is a symmetric
$n\times n$ matrix with zero row and column sums, where
 the $(i,j)$ off-diagonal
entry is given by  $-w_{(i,j)}$ if $(i,j)$ is an edge of $G$
and $0$ otherwise. The $i$th diagonal entry
is given by the weighted degree of vertex $i.$

If $G$ is a connected graph, then   $L_G$   is
a matrix of rank $n-1,$   with   its
kernel  spanned by $\mathbf{1}$ (the vector of all $1$'s).
We let $L_G^+$ denote   the
Moore-Penrose   pseudoinverse   of   $L_G$; this   is a matrix
that   acts as   the  inverse of $L_G$ on $(\ker L_G)^\perp,$
and satisfies $L_G^+ L_G = L_G L_G^+  = I_{n-1},$
where $I_{n-1}$ is the projection   onto the $(n-1)$-dimensional
image   of   $L_G.$

Given the one-to-one correspondence of
graphs and their Laplacians we will often apply algebraic notation
to graphs, with the obvious meaning.
\subsection{Spectral approximation and sparsification}

In this paper we concentrate on symmetric diagonally dominant
matrices. For two matrices $A$ and $B$ of the same  dimension, we
write $A\preceq B$ if $x^T A x  \leq x^T B x$ for all vectors
$x.$ For two graphs $G$ and $H,$ we write $G\preceq H$ if the
Laplacians satisfy $L_G \preceq L_H.$

\begin{definition}
\label{defn:kappaapprox}
We say that a graph  $H$ is a $\kappa$-approximation of
a graph $G$ if $G\preceq H \preceq\kappa G.$
\end{definition}
It is not hard to show that if $H$
is a graph that   $\kappa$-approximates a graph $G$
then   we have
\begin{equation} \label{eq:inverse}
 \frac{1}{\kappa} L_G^+ \preceq L_H^+  \preceq L_G^+
\end{equation}

\begin{definition}
Given a graph $G,$
we say that a (sparser) graph $H$ is a $1\pm\epsilon$ {\bf spectral sparsifier}
of $G$   if
\begin{equation} \label{eq:1pmepsilon}
(1-\epsilon) G \preceq H \preceq (1+\epsilon)  G.
\end{equation}
\end{definition}

It is easy to see that if $H$ is a $1\pm\epsilon$ spectral
sparsifier of $G$ then $\frac{1}{1-\epsilon} H$
is a graph that $\frac{1+\epsilon}{1-\epsilon}$-approximates
$G.$ By the definition, it is also easy
to verify {\bf transitivity}. If $G_1$ is
a $1\pm \epsilon_1$ sparsifier of $G$ and $G_2$ is
a $1\pm \epsilon_2$ of $G_1$ then $G_2$ is a
$(1\pm \epsilon_1)(1\pm \epsilon_2)$ sparsifier of $G$.

\subsection{Graphs as resistive electrical networks}

We can consider our graph $G$ as  an electrical network of
nodes (vertices) and wires (edges), where edge $e$
has resistivity  of $w_e^{-1}$ Ohms.

In this context it is very useful to give another definition of the Laplacian $L_G$,
in terms of its incidence matrix $B_G$. To define $B_G$, fix
an arbitrary orientation for each edge in $G$. For a vertex $i$
let $\chi_i$ be its ($n\times 1$)
characteristic vector,   with a $1$ at the $i$th entry
and $0$'s everywhere else.  Let $e=(i,j)$ be an edge
and  define  $b_e  = \chi_i - \chi_j.$
Then $B_G$ is the $m\times n$ matrix whose $e$th row
is the vector $b_e$.  Let $W_G$ be the $m\times m$ diagonal matrix
whose $e$th diagonal entry is $w_e$.
With these definitions, it is easy to verify that
$$
    L_G = B_G^T W_G B_G = \sum_{e\in G} w_e b_e b_e^T.
$$

For notational convenience,  we will drop
the subscripts on $L_G,$ $B_G,$ and $W_G$ when the
graph we are dealing with is clear from context.

Going back to the electrical analogy,  the \emph{effective resistance}
between vertices $i$ and $j$, denoted by $R^G(i,j)$ or $R^G(e)$ when $(i,j)$
is an edge  $e$, is the voltage difference that has to be applied between $i$ and $j$
in order to drive one unit of external current between the two vertices.
Algebraically it is given by
\begin{equation} \label{eq:efres}
R^G(i,j) = (\chi_i - \chi_j)^T  L_G^+ (\chi_i  - \chi_j)
\end{equation}

The above equation  allows us to apply  \eqref{eq:inverse} and
see that
\begin{equation} \label{eq:approxres}
G \preceq H \preceq \kappa G \Rightarrow (1/\kappa) R^G(e)\leq R^H(e) \leq R^G(e).
\end{equation}

The
definition of the effective resistance for $(i,j)$ in
\eqref{eq:efres} shows directly that it can be computed
by solving the system $L_G x = (\chi_i  - \chi_j)$. In light
of this, \eqref{eq:approxres} will be
of {\em central importance} in our
proofs. Informally, it states that
if $H$ is a $\kappa$-approximation of $G,$ then
the effective resistance
of any edge in $G$ can be approximated  by the
effective resistance of the same edge in $H$, which  can be
done  by solving the system $L_H x = (\chi_i  - \chi_j)$.
This will allow us to construct special approximations
$H$ for which solving with $L_H$ is easier than with $L_G$.

\subsection{Low-stretch subgraphs, spine-heavy graphs, and SDD solvers}

Let $S$ be a graph on the same vertex set as a graph $G$.
Let $e=(i,j)$ be an edge of $G$.
If $p$ is a path $e_1,e_2,\ldots,e_\nu$ between $i$ and $j$ in $S$
we say that the \emph{stretch of $e$ over $p$} is
$\st_p(e) := w_e \sum_{i=1}^\nu w_{e_i}^{-1},$ i.e.\
the weight of $e$ multiplied by the sum of inverse weights
of tree edges on the path from $i$ to $j.$
If ${\cal P}(e)$ is the set of all paths between $i$ and $j$ in $S$
we define
$$
    \st_S(e) = \min_{p\in {\cal P}(e)} \st_p(e).
$$
We will use the term \emph{stretch of $e$ over $S$} for $\st_S(e).$
The definition is simpler when $S$ is a tree. In
this case there is a unique path between the endpoints of $e$.
We denote by $\st_S(G)$  the sum of stretches in $S$ of all
edges of $G,$ i.e.\ $$\st_S(G) = \sum_{e\in G} \st_S(e).$$

It is known that every graph $G$ has a spanning tree $T$ with
$\st_T(G) = O(m \log n \log \log n)$,
known as a {\bf low-stretch tree}. The tree
can computed in $O (m \log n \log \log n)$ time
\cite{AbrahamN12}. Because these guarantees are still
open to improvement we will state our results with
respect to two parameters: We will
denote by $\tff$ the factor in excess of $O(m\log n)$
in the time required for computing 
a low-stretch tree on a graph with $m$ edges, via the
algorithm in \cite{AbrahamN12}.
Similarly, we will denote by 
$\sff$ the factor in excess of $O(m\log n)$
in $\st_T(G)$ provided by the same algorithm.
That is, as noted above, the best current guarantees are $\sff=O(\log\log n)$
and $\tff = O(\log \log n)$.

We call a graph
{\bf spine-heavy} if it has a spanning tree
with  $\st_T(G) = O(m / \log n)$. Given a graph
$G$ we can compute a
spine-heavy graph $H$ that $O(\sff \log^2 n )$-approximates
it  by computing a low-stretch
 tree and then scaling up
 the weights of tree edges in $G$
by the  $O(\sff \log^2 n)$ factor. This is summarized in the
following lemma.

\begin{lemma} \label{lem:scaleup}
For every graph $G$ with $n$ vertices there is a spine-heavy graph $H$ that
 $O(\sff \log^2 n)$-approximates $G$. The graph $H$ can be constructed in time dominated by the
computation of a low-stretch tree for $G$.
\end{lemma}

Finally we state a lemma that summarizes the recent
work on fast SDD solvers \cite{kmp11}.

\begin{lemma} \label{lem:solver}
Let $A$ be an SDD matrix. There is a symmetric operator $\tilde{A}_\delta$
such that
$$
    (1-\delta) A \preceq \tilde{A}_\delta \preceq (1+\delta)A
$$
and that for any vector $b$,
the vector $\tilde{A}^+_\delta b$ can be evaluated
in $O(\tff m \log n + {\sff}    m  \log n \log(1/\delta))$ time. Moreover,
if $A$ is the Laplacian of a spine-heavy graph and its low-stretch
tree is given, then $\tilde{A}^+_\delta b$ can be evaluated in
$O(m \log(1/\delta))$ time.
\end{lemma}

\subsection{Sampling for sparsification}

In a remarkable work, Spielman and Srivastava \cite{effres} analyzed a spectral sparsification algorithm based on a simple sampling procedure. The procedure will be central in our algorithms and we review it here.  It takes as input a weighted graph $G$ and frequencies $p'_e$ for each edge $e$. These frequencies are normalized to probabilities $p_e$ summing to $1$.  It then picks in $q$ rounds exactly $q$ {\em samples}, which are weighted copies of the edges.  The probability that given edge $e$ is picked in a given round is $p_e$. The weight of the corresponding sample is set so that the expected weight of the edge $e$ after sampling is equal to its actual weight in the input graph. The details are given in the following pseudocode.

\begin{algo}[h]
\qquad

\textsc{Sample}
\vspace{0.05cm}

\underline{Input:} Graph $G=(V,E,w)$, $p':E \rightarrow {\mathbb R}^+$.

\underline{Output:} Graph $G'=(V,{\cal L}, w')$.
\vspace{0.2cm}

\begin{algorithmic}[1]
\STATE{$t:= \sum_e p'_e$}
\STATE{$q:= C_s  t \log{t} /\epsilon^2$} {\footnotesize ~~(* $C_S$ is an explicitly known constant  *)}
\STATE{$p_e:= {p'_e}/{t}$}
\STATE{$G':= (V,{\cal L},w')$ with ${\cal L}=\emptyset$}
\FOR{$q$ \mbox{times}}
\STATE{{\small Sample one $e \in E$ with probability of picking $e$ being $p_e$} }
\STATE{{\small Add $l$, a sample of $e$,
to ${\cal L}$ with weight $w'_l = w_e/p_e $}}
\ENDFOR
\STATE{For all $l\in {\cal L}$, let $w'_{l}: = w'_l/q$}
\RETURN{$G'$}
\end{algorithmic}
\end{algo}

Spielman and Srivastava analyzed the case when $p'_e = w_eR^G(e)$, where $R^G(e)$ is the effective resistance
of $e$ in $G$. The following generalization characterizes the quality of $G'$ as a spectral sparsifier for $G$. It is shown in \cite{kl} and it was originally proved with a weaker success guarantee in \cite{kmplinsolve}.

\begin{theorem}\label{th:oversampling}
\textbf{(Oversampling)} Let $G=(V,E,w)$ be a graph. Assuming that $p'_e \geq w_eR^G(e)$ for each edge $e\in E$ the graph $G' = \textsc{Sample}(G, p')$ is a $(1\pm \epsilon)$ sparsifier of $G$ with probability at least $1-1/n^2$.
\end{theorem}

\subsection{Incremental spectral sparsifiers}
\label{sec:incrspars}
In \cite{kmplinsolve}, Koutis, Miller, and Peng
asked whether   they   could get anything useful out of the
Spielman-Srivastava construction,
without having to use effective resistances whose computation
requires the solution of  linear system solvers. The oversampling
Theorem \ref{th:oversampling} offers a possibility as
long as we are able to compute, more efficiently,
upper bounds to the quantities $w_e R^G(e)$.

The idea
in \cite{kmplinsolve} was to use spanning trees and
use as an upper bound the stretch of an edge over the tree.
The idea generalizes readily to spanning subgraphs. The reason is that
for any subgraph $S$, we have $\st_S(e) \geq w_e R^{G}(e)$ by
Rayleigh's theorem.

In the case $S$ is a tree, we can
compute the stretches
of all the edges  in $O(m)$ time (using an offline lowest
common ancestor algorithm  \cite{Tarjan79, GabowT83}).
For one of our results we will also take $S$ to be
a so-called $O(\log n)$-spanner of the graph. In this case, by definition
of spanner, we have $\st_S(e) = O(\log n)$ for all $e$;
we will use directly these estimates without further computations.

In order to  get a useful approximation via oversampling, we need the
number of edges in the resulting object to be
significantly smaller
than $m.$  The intuition is that the lower we get $\st_S(G)$ the
fewer samples we need to take, since  the sum of the
 overestimates of the probabilities
defines the number of samples. In \cite{kmplinsolve}
this intuition is coupled with a scaling-up technique,
where instead of $G$ we apply oversampling on
a graph $H$  that
is the same as $G$  except that the weights of the edges of $S$
are scaled up by a suitably chosen factor $\kappa.$
By doing that, we lose
a factor of  $\kappa$ in the  approximation guarantee but
at the same time we in fact lower the total stretch
by a $\kappa$ factor, and, for suitable
$\kappa,$ we will have
have  a small enough
edge count in the graph $I$ we output.
We call this graph the
\emph{incremental spectral sparsifier}.

We summarize  the result in Theorem \ref{thm:incrsparsifier} below, which is
an adaptation of a corresponding Theorem in \cite{kmp11}, coupled
with the stronger probabilistic guarantee of success of Theorem \ref{th:oversampling},
and slightly generalized to handle general subgraphs rather than only spanning trees.

\begin{theorem}
\label{thm:incrsparsifier}
Let $G$ be a graph and $S$ be a spanning subgraph of $G$.
Assume that $\st_S(e)$ is known for each edge $e$ of $G$.
Then, there is an algorithm for constructing a graph
$I$ such that, given $\kappa$:
\begin{itemize}
\item
$H\preceq I\preceq 2 H$, where $H=G+\kappa S$.
\item
$I$ has  $n- 1+ O( \st_S(G) \log n /\kappa  )$ edges
\end{itemize}
assuming $\st_S(G)$ is polynomially bounded.
The algorithm succeeds with high  probability and runs in in  $O(m + \st_S(G) \log n /\kappa)$ time.
\end{theorem}

%% file: sparsification.tex
\section{The general case: An $O(\sff m\log^2 n)$ time algorithm}  \label{sec:sparsification}
\subsection{Estimating effective resistances}

As we discussed above, Spielman and Srivastava \cite{effres}
use the \textsc{Sample} algorithm with $p_e' = w_eR^G(e)$.
For the efficient implementation of their algorithm
they first obtain a different expression for the effective
resistance, via a simple algebraic manipulation:
\begin{eqnarray*}
R^G(i,j) &=& (\chi_i - \chi_j)^T L^+ (\chi_i - \chi_j)\\
&=& (\chi_i - \chi_j)^T L^+ L L^+ (\chi_i - \chi_j)\\
&=&(\chi_i - \chi_j)^T L^+ B^T W^{1/2} W^{1/2} B L^+ (\chi_i - \chi_j)\\
&=& \|W^{1/2}B L^+(\chi_i-\chi_j)\|^2
\end{eqnarray*}
The advantage of this definition is that
it expresses
the effective resistance as the squared Euclidean
distance between  two points, given by the $i$th and $j$th column
of the matrix $W^{1/2}B L^+$.
This new expression still  involves the solution of a linear system
with $L$. The natural idea is to replace $L$ with
an approximation $\tilde{L}_\delta$ satisfying  the properties
described in Lemma \ref{lem:solver}. So instead of
$R^G(i,j)$ we compute the quantities
$\hat{R}^G(i,j) = \|W^{1/2}B \tilde{L}^+_\delta(\chi_i-\chi_j)\|^2$.

Of course, there are still $m$ systems to be solved. To work around
this hurdle, Spielman and Srivastava observe that projecting the
vectors to an $O(\log n)$-dimensional space  preserves the
Euclidean distances within a factor of $1\pm \epsilon/8$, by the Johnson-
Lindenstrauss theorem. Algebraically
this amounts to computing the quantities  $\|Q W^{1/2}B \tilde{L}^+_\delta
(\chi_i-\chi_j)\|^2$, where $Q$ is a properly defined random matrix
of dimension $k\times m$ for $k = O(\log n)$. The authors invoke the result
of Achlioptas \cite{jl}, which states that  one can use
a matrix $Q$ each of whose entries is randomly chosen
in $\{\pm1/\sqrt{k}\}.$

The construction of the   sparsifiers can  can thus   be broken
up into three  steps.
\begin{enumerate}
\item
Compute $Q W^{1/2}B$. This takes time $O(km),$ since $B$ has only two
non-zero entries per row.
\item
Apply the linear operator $\tilde{L}^+_\delta$ to the $k$ columns of the
matrix $(Q W^{1/2}B)^T$, using Lemma \ref{lem:solver}. This gives
the matrix $Z=Q W^{1/2}B \tilde{L}^+_\delta$.
\item
Compute all the (approximate) effective resistances (time $O(k m)$) via the square norm of
the differences between columns of the matrix $Z$. Then sample the edges.
\end{enumerate}

\subsection{The $O(\sff m\log^2 n)$ time algorithm}
Spielman and Srivastava prove that the
approximations $\hat{R}^G(i,j)$ can be used to obtain the sparsifier
if they satisfy
$$
   (1-\epsilon/4) {R}^G(i,j) \leq  \hat{R}^G(i,j) \leq (1+\epsilon/4){R}^G(i,j).
$$
Then they show that this can be satisfied if
$\delta,$
the  accuracy guarantee of
the  linear system solver,  is taken to be
an {\em inverse polynomial} in $n$. Thus their algorithm
is dominated by the second step (the applications of $\tilde{L}^+_\delta$)
and  takes time $O(\tff m\log n + \sff m   \log^3 n  \log(1/\epsilon) )$.

The following lemma shows that in fact it is enough to
take $\delta$ to be a constant.  Furthermore,
our   proof   significantly simplifies
 the corresponding analysis   of \cite{effres}.
\begin{lemma}
\label{lem:solvererror}
For a given $\epsilon$, if $\tilde{L}$ satisfies $(1-\delta)L \preceq \tilde{L} \preceq (1+\delta)L$
where $\delta = \epsilon / 8 $,
then the approximate effective resistance values $\hat R^G(u,v) = \|W^{1/2}B\tilde{L}^+(\chi_u - \chi_v)\|^2$
satisfy:
\begin{align*}
    (1 - \epsilon) R^G(u,v) \leq \hat{R}^G(u,v) \leq (1 + \epsilon) R^G(u,v).
\end{align*}
\end{lemma}

\begin{proof}
We only show the first half of the inequality, as the other half follows similarly.
Since $L$ and $\tilde{L}$ have the same null space, by
 \eqref{eq:inverse}  the given condition is equivalent to:
\begin{align*}
\frac{1}{1+\delta}L^+ \preceq \tilde{L}^+ \preceq \frac{1}{1-\delta}L.
\end{align*}
Since $\frac{1}{1+\delta}L^+ \preceq \tilde{L}^{+}$, we have
\begin{align*}
R^G(u,v) & = (\chi_u - \chi_v)^TL^+(\chi_u - \chi_v) \\
& \leq (1+\delta) (\chi_u - \chi_v)^T\tilde{L}^+(\chi_u - \chi_v) \\
& = (1+\delta) (\chi_u - \chi_v)^T\tilde{L}^+ \tilde{L} \tilde{L}^+ (\chi_u - \chi_v).
\end{align*}

Applying the fact that $\tilde{L}  \preceq (1+\delta)L$
 to the vector
$\tilde{L}^+ (\chi_u - \chi_v)$ in turn gives:
\begin{align*}
R^G(u,v) & \leq (1+\delta)^2 (\chi_u - \chi_v)^T\tilde{L}^+ L \tilde{L}^+ (\chi_u - \chi_v) \\
& = (1+\delta)^2 \|W^{1/2} B \tilde{L}^+ (\chi_u - \chi_v)\|^2 = (1+\delta)^2 \hat{R}^G(u,v)
\end{align*}

The rest of the proof follows from $\frac{1}{(1+\delta)^2} \leq 1-\epsilon/4$ by choice of $\delta$.
\end{proof}

This proves our first theorem.
\begin{theorem}
 There is a $1\pm \epsilon$ sparsification algorithm that runs in $O(\tff m\log n + \sff m\log^2 n \log(1/\epsilon))$ time.
\end{theorem}

\section{The $O(\tff m\log n)$ time algorithm} \label{sec:improve}
\label{sec:mlognalg}

Informally, the oversampling Theorem \ref{th:oversampling} states that
if we use estimates to   the effective resistances,
rather than the true values,  the Spielman-Srivastava scheme still works; but
 in order to produce the sparsifier we have to compensate by taking
more samples. We exploit this in our second Theorem.

\begin{theorem} \label{th:first}
 There is a $(1\pm \epsilon)$-sparsification algorithm
 that runs in $O(\tff m\log n + m\log n \log(1/\epsilon) )$ time and returns a sparsifier
 with $O(\sff n\log^3 n/\epsilon^2)$ edges.
 As a result, we can compute an $(1\pm \epsilon)$-sparsifier with
 $O(n\log n/\epsilon^2)$ edges in
 $O(\tff m\log n+ m\log n \log(1/\epsilon) + \sff^2 n\log^5 n / \epsilon^2)$ time.
\end{theorem}
\begin{proof}
Given the input graph $G$ we construct a spine-heavy graph $H$ that $O(\sff \log^2 n)$-approximates $G$.
The construction can be done in $O(\tff m\log n)$ time, by Lemma \ref{lem:scaleup}. We then run
the Spielman-Srivastava scheme (Section \ref{sec:sparsification}) on $H$ to approximate the effective resistances $R^H(i,j)$ within a factor
of $1\pm \epsilon$. Step 2
of the Spielman-Srivastava scheme runs in $O(m \log n \log(1/\epsilon))$ time on $H$, by Lemma \ref{lem:solver}.
We adjust the approximate effective resistances in $H$ down by a factor of $1+\epsilon$ to
accommodate for the upper side of the error in Lemma \ref{lem:solvererror}. Then, by
 \eqref{eq:1pmepsilon} the calculated approximate effective resistances satisfy
$$ \frac{1}{O(\sff \log^2 n)} R^G(i,j) \leq  \hat{R}^H (i,j) \leq
R^G(i,j).$$
So Theorem \ref{th:oversampling}
applies if we take $p_e'=O(\sff \log^2 n) w_e \hat{R}^H (i,j)$. We have
$$
  \sum_e p_e' = O(\sff \log^2 n) \sum_e w_e \hat{R}^H (i,j) \leq O(\sff \log^2 n) \sum_e w_e R^G(i,j) = O(\sff n \log^2 n).
$$
The last equality follows from the fact that $\sum_e w_e R^G(i,j)=n-1$ for any graph $G$ (e.g. see \cite{effres}).
Hence the total number of samples we need to take in order
to produce an $(1\pm \epsilon)$-sparsifier is  $O(\sff n\log^3 n/\epsilon^2)$.
The second sparsifier is computed by re-sparsifying with the general case algorithm (and appropriate settings for $\epsilon$).
\end{proof}

%% file: improve.tex

\section{The fastest algorithms} \label{sec:fastest}

\subsection{A near-linear stretch tree in $O(m)$ time}

The first problem in trying to accelerate the algorithm of the
previous Section lies in the computation of the low-stretch tree;
known algorithms for the task take time at least $O(m\log n)$.
To work around this problem we will trade-off stretch for time.
In the remainder of this subsection we show that we can in $O(m)$ time
produce a spanning tree with a slightly weaker stretch guarantee, namely
$\st_T(G) = O( \sff^2 m\log^3 n  )$. The construction goes through
the  computation of incremental {\em cut} sparsifiers,
which may be of independent interest.

\medskip
\noindent \textbf{A cut-based characterization of stretch.} We start by giving a simple alternative characterization of the stretch of a graph over a tree; for this we will need some notation.  Let $G=(V,E_G)$ be a graph and $T=(V,E_T)$ be a spanning tree of $G$. Every edge $e\in E_T$ defines in the obvious way a partition of $V$ into two sets $V_e$ and $V-V_e$. Indeed,
removing the edge disconnects the tree, and the partition is formed by the
vertices in the  two connected components; we arbitrarily let $V_e$ be
the vertices in one of them.
 Let $\ccap_G(V_e,V-V_e)$ be the total weight of the edges in $G$ with endpoints in $V_e$ and $V-V_e$. We have
\begin{equation} \label{eq:alt_stretch}
    \st_T(G) =  \sum_{e \in E_T} w_e^{-1} \ccap_G(V_e,V-V_e).
\end{equation}

To see why, let us go back to the definition of $\st_T(G)$ given in Section \ref{sec:background}. We have
$$
    \st_T(G) = \sum_{e' \in E_G} w_{e'} \sum_{e \in p_T(e)} w_e^{-1},
$$
where $p_T(e)$ is the unique path in $T$ between the two endpoints of $e$. It is then clear that $\st_T(G)$ is a sum of terms of the form
$w_{e'}w_{e}^{-1}$
for $e' \in E_G$ and $e\in E_T.$ Instead of grouping the terms  with respect to $e'\in E_G$ as is customary in the above definition, we group them with respect to $e \in E_T$. It can be seen that for any fixed $e\in E_T$,
the term $w_{e}^{-1}$ appears as a factor multiplying $w_{e'}$ for each edge $e' \in E_G$ that has its two endpoints in $V_e$ and $V-V_e$, precisely when $p_T(e')$ has to use $e$. This directly gives us the alternative characterization in Equation \ref{eq:alt_stretch}.

\medskip
\noindent \textbf{Incremental cut sparsifiers.} Inspired by the incremental spectral sparsifiers of \cite{kmplinsolve}, we give an analogous construction of incremental cut sparsifiers, i.e.\ sparsifiers that approximately preserve cuts but are only mildly sparser relative to the input graph. We will say that a graph $H$ is a $\tau$-cut approximation of $G$ if  for all $S\subseteq V$ we have
$$
   \tau \cdot \ccap_H(S,V-S) > \ccap_G(S,V-S).
$$

We claim the following Lemma.
\begin{lemma} \label{th:cut-approximation}
 There is an $O( \tff^2 \log^3 n  )$-cut approximation $H$ of a graph $G$ with $O(m/(\tff \log n) + n)$ edges. The graph $H$ can be computed in $O(m)$ time with high probability.
\end{lemma}
The algorithm and its proof is based on Karger's earlier work on cut sparsification \cite{Karger98}. Before we proceed with it, we review necessary definitions and a Lemma from \cite{Karger98}. In the following context graphs are allowed to have multiple edges.

A graph is $k$-connected if the value of each cut in $G$ is at least $k$. A $k$-strong component is a $k$-connected vertex induced subgraph of a graph. The strong connectivity of an edge $e$, denoted $k_e$, is the maximum value of $k$ such that a $k$-strong component contains $e$. Finally, a graph $G$ is said to be $c$-smooth if for every edge $e$, $k_e \geq c w_e$.

Now let $G(p)$ be the subgraph of $G$ resulting by keeping each edge of $G$ with probability $p$.
\begin{lemma} \textnormal{\cite{Karger98}} \label{th:uniform}
Let $G$ be a $c$-smooth graph.  Let $p = \rho_{\epsilon}/c$ where $\rho_\epsilon = O(\log n/\epsilon^2)$. Then with high probability every cut in $G(p)$ has a value in the range $(1\pm \epsilon)$ times its expectation (which is $p$ times its original value).
\end{lemma}

We now proceed with the algorithm and its proof.

\begin{proof} (of  Lemma \ref{th:cut-approximation})
We can assume without loss of generality that the edge weights in $G$ are integers. We first find a max-weight spanning tree $T$ of $G$; since the weights are integer this can be done in $O(m)$ time. We then form an intermediate graph $G'$ by multiplying
the weight of every edge in $T$ by $\lceil \tff \log^2 n \rceil$; let $T'$ be $T$ with the scaled up weights.

Consider now an edge $e$ in $G'$ which is not in $T'$. Let $p_{T'}(e)$ be the path connecting the endpoints of $e$ in $T'$. Because $T$ is a maximum-weight spanning tree in $G$ every edge along $p_T(e)$ has weight at least $w_e$ in $G$. Therefore the subgraph of $G'$ induced by the vertices in $p_{T'}(e)$ is $(w_e  \lceil \tff \log^2 n \rceil)$-strong. Therefore the connectivity $k_e$ satisfies $k_e\geq  ( \lceil \tff \log^2 n \rceil) w_e$.

In order to be able to apply Lemma \ref{th:uniform} we will modify $G'$ to be a multigraph, by viewing  each $e\in T'$ as  $ \lceil \tff \log^2 n \rceil$ parallel edges of weight $w_e$. Under this definition, the connectivity of each such parallel edge $e'$ trivially satisfies $k_{e'} \geq ( \lceil \tff \log^2 n \rceil) w_e$. The same holds for all other edges of $G'$ as shown above. It follows that the multi-graph $G'$ is $O(\tff \log^2 n)$-soft.

We can now apply Lemma \ref{th:uniform}, setting $\epsilon = 1/2$ to form $G(p)$ for $p = O(1/(\tff \log n) )$. We get that $(2/3)G(p)$ is an $O(\tff \log n)$-cut approximation for $G'$. By an easy transitivity argument, the graph $H=2G(p)/(3 \lceil \tff \log^2 n \rceil)$ is then an  $O(\tff^2 \log^3 n)$-cut approximation of $G$. The claim about the number of edges of $H$ follows by application of Chernoff's inequality.
\end{proof}

\medskip
\noindent \textbf{Computing a low-stretch tree faster.} We conclude this subsection with the main Lemma.

\begin{lemma} \label{th:lowish-stretch}
Given a graph $G$, a spanning tree $T$ such that $\st_T(G) = O(\tff^2 m \log^3 n)$ can be computed in $O(m)$ time.
\end{lemma}
\begin{proof}
We first produce in $O(m)$ time the graph $H$ of Lemma \ref{th:cut-approximation}.
We then apply the low-stretch algorithm of \cite{AbrahamN12} on graph $H$ to get a spanning tree $T'$ of $H$;
given the number of edges in $H$ this step takes $O(m)$ time as well, while $T'$ satisfies
$
   \st_{T'}(H) = O(m).
$
Notice now that the edge weights in $H$ are by definition smaller than those in $G$.  Therefore, the corresponding
tree $T$ in $G$ (i.e. the tree with the original weights) satisfies
$$
   \st_{T}(H) = O(m).
$$
Given the number of edges in $H$ the low-stretch algorithm runs in $O(m)$ time. Then using the definition of cut approximation, and equation \ref{eq:alt_stretch} we get that
\begin{eqnarray*}
   \st_T(G) & = &  \sum_{e \in E_T} w_e^{-1} \ccap_G(V_e,V-V_e).  \\
   & \leq &  O(\tff^2 \log^3 n ) \sum_{e \in E_T} w_e^{-1} \ccap_H(V_e,V-V_e)\\
   & = &  O(\tff^2 \log^3 n )\cdot \st_T(H) \\
   & = & O(\tff^2 m \log^3 n).
\end{eqnarray*}
\end{proof}

\subsection{Leveraging an approximate sparsifier} The purpose of this
section is to show that
if we are given an {\em approximate sparsifier} for a graph $G$
we can efficiently produce
 a sparsifier for $G$.  So far we have been using only low-stretch \emph{subgraphs} of $G$ to get
approximations to the effective resistances in $G$.  A key to our fastest algorithm
is the realization that we can find low-stretch trees that are not necessarily subgraphs     of
the given graph $G$; the total stretch will actually be only a near-linear function of $n$.
This is based on the following Lemma.

\begin{lemma} \label{th:stretchdom}
Let $H' \preceq H $. Then for any tree $T$
$$
    \st_{T}(H') \leq \st_{T}(H).
$$
\end{lemma}
\begin{proof}
 Let $A^+$ denote the Moore-Penrose pseudoinverse of a matrix $A$ and $\lambda_i(A)$ denote the $i^{th}$ largest
 eigenvalue of $A$.  By Spielman and Woo \cite{SpielmanW09} we know that for any graph $G$
 $$
      \st_{T}(G) = \trace(L_GL_T^+) = \sum_i \lambda_i (L_GL_T^+)
 $$
 Because $L_G$ and $L_T$ have the same null space (the constant vector), we can write
 $$
    \lambda_i(L_GL_T^+) = \lambda_i (L_T^{+/2} L_GL_T^{+/2}).
 $$
 Notice now that we have
 $$
     H' \preceq H  \Rightarrow (L_T^{+/2} L_{H'} L_T^{+/2}) \preceq (L_T^{+/2} L_HL_T^{+/2}).
 $$
 This follows easily by definition. It is
  also easy to prove (see for example \cite{Koutis-thesis}, Ch. 6.1) that
 $$
      A \preceq B \Rightarrow \lambda_i(A) \leq \lambda_i(B).
 $$
 Hence
 \begin{eqnarray*}
      \lambda_i (L_T^{+/2} L_{H'} L_T^{+/2}) & \leq  & \lambda_i(L_T^{+/2} L_{H} L_T^{+/2}) \Rightarrow \\
      \trace(L_T^{+/2} L_{H'} L_T^{+/2}) & \leq & \trace(L_T^{+/2} L_{H} L_T^{+/2}) \Rightarrow \\
      \st_T(H') & \leq & \st_T(H).
 \end{eqnarray*}
\end{proof}

We are now ready to prove the main Lemma in this subsection. To avoid confusion we will use
$m_H$ to denote the number of edges of a graph $H$.

\begin{lemma} \label{th:leveraging}  (\textbf{Leveraging})
 Let $H'$ be a $\kappa$-approximation of $H$. Suppose
 we are given $\tilde{H}'$, a $4$-approximation of $H'$ with $O(n\log n)$ edges.
 Then we can construct an $(1\pm \epsilon)$-approximation of $H$ with $O(\sff \kappa n \log^3 n / \epsilon^2)$ edges
 in $O(m_H + \tff n\log^2 n)$ time. We can also construct
 am $(1\pm \epsilon)$-sparsifier of $H$ with $O(n\log n/\epsilon^2)$ edges in
 $O(m_H + \sff^2  \kappa n\log^\up n/\epsilon^2)$ time.
 \end{lemma}
\begin{proof}
 We compute a low-stretch spanning tree $T$ of $\tilde{H}'$ in $O(\tff n\log^2 n)$
 time. Because $\tilde{H}'$ has $O(n\log n)$ edges we have
 $$
     \st_T(\tilde{H}') = O(\sff n\log^2 n).
 $$

 We then compute in $O(m_H)$ time the effective resistance $R^T(e)$ of each edge $e$ of $H$ over $T$.
 This can be done in $O(m_H)$ time using off-line LCA algorithms
by Gabow and Tarjan \cite{Tarjan79, GabowT83}.
 Since $\tilde{H}' \preceq 4H'$ and $H' \preceq H$ we can apply \eqref{eq:1pmepsilon} twice and get
 $$
     R^H(e) \leq R^{H'} (e) \leq 4 R^{\tilde{H}'}(e).
 $$
 By Rayleigh's monotonicity theorem (e.g. see \cite{kmplinsolve}) we get that:
 $$
    R^{\tilde{H}'}(e) \leq R^T(e).
 $$
 Hence setting $p_e' = 4 w_eR^{T}(e)$ in Theorem \ref{th:oversampling} allows us to sparsify $H$.
 To get a $(1\pm \epsilon)$-approximation of $H$, the number of samples we need to take is
 $O(q\log q/\epsilon^2)$ where $q= \sum_e 4w_e R^{T}(e)=4\st_T(H).$
 Since $H' \preceq \tilde{H}'$ and $H \preceq \kappa H'$, we apply Lemma \ref{th:stretchdom}
 twice and we get that
 $$
      \st_T(H) \leq \kappa \cdot \st_T(H') \leq \kappa \cdot \st_T(\tilde{H}') = O(\sff \kappa n\log^2 n).
 $$
 This proves the first claim.  The second claim follows via picking the appropriate $\epsilon$ and
 re-sparsifying the first sparsifier with the general
 case sparsification algorithm.
\end{proof}

\subsection{The $O(m)$ and $O(m\log \log n)$ time algorithms}

We are now ready to prove our main claims.

\begin{theorem} \label{th:linear}
 Given a graph $G$ we can compute  an $(1\pm \epsilon)$-spectral sparsifier of $G$ with $O(n\log n/\epsilon^2)$
 edges in $O(m+\tff^2 \sff^3 n\log^{10} n/\epsilon^2)$ time.
\end{theorem}
\begin{proof}
Let $T$ be the spanning tree of $G$ provided by Lemma \ref{th:lowish-stretch}.
Let $H$ and $I$ be the two graphs from the application of Theorem \ref{thm:incrsparsifier}
on $G$ with respect to $T$ and $\kappa =O(\tff^2 \sff \log^5 n)$.
 We  use our  algorithm from Theorem~\ref{th:first} to 
sparsify $I$ and by transitivity we get a $4$-approximation $H'$ of
$H$ with $O(n\log n)$ edges in $O(m+ \sff^2 n\log^5 n)$ time.
Because $H$ is an $O(\tff^2 \sff \log^5 n)$-approximation
of $G$ we can `leverage' its sparsifier $H'$ to compute a $(1\pm\epsilon)$-sparsifier of $G$
via Lemma \ref{th:leveraging} (with $\kappa = O(\tff^2 \sff \log^5 n)$) in
$O(m+  \tff^2 \sff^3 n\log^{10} n/\epsilon^2)$ time.
\end{proof}

\begin{theorem} \label{th:unweighted}
 Given an unweighted graph $G$ we can compute  an $(1\pm \epsilon)$-spectral sparsifier of $G$ with $O(n\log n/\epsilon^2)$
 edges in $O(m+\sff^3 n\log^{8} n/\epsilon^2)$ time.
\end{theorem}
\begin{proof}
Let $S$ be an $O(\log n)$-spanner of $G$, as discussed in
Section~\ref{sec:incrspars}. Such
an $S$ with $O(n)$ edges can be computed in $O(m)$ time \cite{halperin96}.
Let $H$ and $I$ be the two graphs from the application of Theorem \ref{thm:incrsparsifier}
on $G$ with respect to $S$ and $\kappa =O(\sff \log^3 n)$. We then use our  algorithm
from Theorem~\ref{th:first} to
sparsify $I$ and by transitivity we get a $4$-approximation $H'$ of
$H$ with $O(n\log n)$ edges in $O(m+ \sff^3 n\log^5 n)$ time.
Because $H$ is an $O(\sff \log^3 n)$-approximation
of $G$ we can `leverage' its sparsifier $H'$ to compute a $(1\pm\epsilon)$-sparsifier of $G$
via Lemma \ref{th:leveraging} (with $\kappa = O(\sff \log^3 n)$) in
$O(m + \sff^3 n\log^{8} n/\epsilon^2)$ time.
\end{proof}

\begin{theorem} \label{th:almostlinear}
 Given a graph $G$ we can compute an $(1\pm \epsilon)$-spectral sparsifier of $G$ with $O(n\log n/\epsilon^2)$
 edges in $O((m +\sff^2 n\log^5 n/\epsilon^2)\log \log n)$ time.
\end{theorem}
\begin{proof}
Let $T$ be the spanning tree of $G$ provided by Lemma \ref{th:lowish-stretch}.
Let $H = G+ O(\tff^2 \sff \log^5 n) T$.
We construct a sequence of graphs,
$$
    H = H_0,H_1,\ldots,H_t= G
$$
where $H_i = G+ O(\tff^2 \sff \log^5 n) T/2^i$, for $i=1\ldots,t-1$ with some appropriate $t=O(\log \log n)$.
Notice that all graphs $H_i$ have $m$ edges, so this takes $O(m\log \log n)$ time.
For $i=0,\ldots ,t-1$, let $\tilde{H_i}$ denote a 4-approximation of $H_i$
with $O(n\log n)$ edges.  It is clear that $\tilde{H_0}$ can be computed in $O(m)$ time.
Provided now that we have $\tilde{H}_j$ we can apply Lemma \ref{th:leveraging} (with $\kappa=2$, $\epsilon=1/2$ and a proper scaling of the $(1\pm \epsilon)$-sparsifier)
to get a 4-approximation $\tilde{H}_{j+1}$ in $O(m+\sff^2 n\log^\up n)$ time.
From $H_{t-1}$ we produce a $(1\pm \epsilon)$-sparsifier of $H_t=G$
again by Lemma \ref{th:leveraging}, using the desired value for $\epsilon$.
Because we apply the algorithm of Lemma \ref{th:leveraging} $O(\log \log n)$
times, we get the claimed running time.
\end{proof}

%% file: discussion.tex
\section{Final Remarks}

The original algorithm of Spielman and Teng remains the only known
combinatorial sparsification algorithm that does not rely on solving systems.
Designing a spectral sparsification algorithm
that does not depend  on a linear system solver and
 that outputs a very sparse graph with
$O(n\log n)$ or $O(n\log^2 n)$ edges is a challenging open problem.
Since achieving this may prove to be difficult, an alternate approach
could be algorithms that compute  very sparse $\kappa$-approximations
for small values of $\kappa$.
Such algorithms could play a significant role in the development of
more practical SDD solvers.

\section{Acknowledgments}
We would like to thank Jonathan Kelner and Gary Miller
for useful discussions and the anonymous referees for carefully reading
an earlier version of this paper.
Ioannis Koutis is supported by NSF CAREER award CCF-1149048.
Alex Levin is supported by a National Science Foundation
graduate fellowship. Richard Peng was at Microsoft Research New England
for part of this work and is supported by a Microsoft Research Fellowship.
He is also partially supported by the National Science Foundation under grant number CCF-1018463.